\documentclass[letterpaper, 10 pt, conference]{ieeeconf}  

\IEEEoverridecommandlockouts                              
\usepackage{color}
\usepackage{amsmath} 
\usepackage{amssymb}  
\usepackage{tikz}
\usepackage{cite}
\usepackage{pgfplots}
\usepackage{hyperref}
\pgfplotsset{
	compat=newest,
	/pgfplots/legend image code/.code={%
		\draw[mark repeat=2,mark phase=2,#1] 
		plot coordinates {
			(0cm,0cm) 
			(0.1cm,0cm)
			(0.2cm,0cm)
			(0.3cm,0cm)
			(0.4cm,0cm)%
		};
	},
}
\newlength\hohe
\newlength\breite

\newcommand{\regret}{\mathcal R_T}
\newcommand{\norm}[1]{\left\lVert#1\right\rVert}
\newcommand{\MAS}{\mathcal{O}_\infty}
\newcommand{\RPI}{\mathcal{P}^\lambda}
\newcommand{\refeq}[1]{\overset{#1}{=}}
\newcommand{\refleq}[1]{\overset{#1}{\leq}}

\newcommand{\refin}[1]{\overset{#1}{\in}}
\newcommand{\refsubseteq}[1]{\overset{#1}{\subseteq}}
\newcommand{\refrarrow}[1]{\overset{#1}{\Rightarrow}}
\newcommand{\reflrarrow}[1]{\overset{#1}{\Leftrightarrow}}
\newcommand{\overbar}[1]{\mkern 1.5mu\overline{\mkern-1.5mu#1\mkern-1.5mu}\mkern 1.5mu}

\newtheorem{defn}{Definition}
\newtheorem{thm}{Theorem}
\newtheorem{ass}{Assumption}
\newtheorem{lem}{Lemma}

\title{\LARGE \bf
Robust Control of Constrained Linear Systems using Online Convex Optimization and a Reference Governor
}

\author{Marko Nonhoff, Mohammad T. Al Torshan and Matthias A. Müller
\thanks{This work was supported by the Deutsche Forschungsgemeinschaft (DFG, German Research Foundation) - 505182457.}
\thanks{The authors are with the Leibniz University Hannover, Institute of Automatic Control, Germany (email: \{nonhoff,mueller\}@irt.uni-hannover.de, al.torshan@stud.uni-hannover.de).}}%

\newcommand\copyrighttext{%
	\footnotesize \textcopyright2024 IEEE. Personal use of this material is permitted.  Permission from IEEE must be obtained for all other uses, in any current or future media, including reprinting/republishing this material for advertising or promotional purposes, creating new collective works, for resale or redistribution to servers or lists, or reuse of any copyrighted component of this work in other works.}
\newcommand\copyrightnotice{%
	\begin{tikzpicture}[remember picture,overlay]
		\node[anchor=south,yshift=5pt,xshift=0pt,fill=white] at (current page.south) {\fbox{\parbox{\dimexpr.9\textwidth-\fboxsep-\fboxrule\relax}{\copyrighttext}}};
	\end{tikzpicture}%
}

\begin{document}

\maketitle
\thispagestyle{empty}
\pagestyle{empty}

\begin{abstract}
This article develops a control method for linear time-invariant systems subject to time-varying and a priori unknown cost functions, that satisfies state and input constraints, and is robust to exogenous disturbances. To this end, we combine the online convex optimization framework with a reference governor and a constraint tightening approach. The proposed framework guarantees recursive feasibility and robust constraint satisfaction. Its closed-loop performance is studied in terms of its dynamic regret, which is bounded linearly by the variation of the cost functions and the magnitude of the disturbances. The proposed method is illustrated by a numerical case study of a tracking control problem.
\end{abstract}

\section{INTRODUCTION} \label{sec:intro}
\copyrightnotice

In recent years, controllers for dynamical systems based on the online convex optimization (OCO) framework have gained significant interest in the literature. Originally, the OCO framework was developed to handle optimization problems which feature time-varying and a priori unknown cost functions (compare, e.g., \cite{Shalev12,Hazan2016} and the references therein for an overview). More specifically, at each time instance~$t$ an OCO algorithm computes an action $u_t$ based solely on previous actions and cost functions, but has no access to current or future cost functions. Since such time-varying and a priori unknown cost functions also arise in a variety of control applications, for example due to renewable energy generation, time-varying loads or a priori unknown energy prices \cite{Tang17,Menta18} or in tracking control, when the desired trajectory to be tracked is computed online itself \cite{Zheng20,Tsiamis24}, numerous controllers based on the OCO framework have been proposed in the literature, see, e.g., \cite{Nonhoff24,Agarwal19,Shi2020,Li21,Hazan2022,Lin2023,Karapetyan2023} and the references therein. Most commonly, these algorithms aim to track the optimal steady states of the controlled system, which are a priori unknown and time-varying themselves, because they depend on the a priori unknown time-varying cost functions.

In addition to the ability to operate in dynamic environments with time-varying and a priori unknown cost functions, safety in the form of constraint satisfaction guarantees is a crucial requirement in many practical applications. Such constraints are ubiquitous in practice, where, e.g., actuator limitations and physical restrictions of the system under control have to be taken into account. Similarly, robustness to exogenous disturbances is of paramount importance in practice because of their potentially harmful effect on closed-loop performance and constraint satisfaction guarantees. However, OCO-based algorithms that are able to handle both, constraints and exogenous disturbances, are much scarcer in the literature \cite{Nonhoff24,Li21}. Furthermore, in \cite{Li21}, the setting is limited to disturbance rejection, i.e., only stabilization of the origin is considered in contrast to tracking of time-varying steady states, whereas in \cite{Nonhoff24} at least a linear program has to be solved at each time step, increasing the computational burden of the proposed algorithm.

Another closely related line of research is so-called feedback optimization, where the controller aims to drive a dynamical system to the solution of a (possibly constrained and time-varying) optimization problem, see, e.g., \cite{Simonetto2020,Hauswirth21,Lawrence2021,Chen23} and the references therein. Whereas exogenous disturbances are commonly considered in this framework, constraints are typically only enforced for the optimal steady state. In particular, satisfaction of constraints on the state of the system can generally not be guaranteed at all times.

In this work, we build on and extend the results from \cite{Nonhoff22b}. We propose a framework for robust control of dynamical systems that is subject to time-varying and a priori unknown cost functions as well as pointwise in time state and input constraints. To do so, we make use of a conceptually simple approach: First, we design an OCO algorithm that computes a reference signal for the controlled system, which aims to track the optimal steady state, thereby handling the time-varying and a priori unknown cost functions. Second, we employ a reference governor (RG) that modifies the OCO algorithm's reference in order to enforce constraint satisfaction. Reference governors are add-on schemes that modify the reference to a well-designed closed-loop system in order to avoid constraint violation, see, e.g., \cite{Garone17} for a recent survey on the topic. Compared to our previous work \cite{Nonhoff22b}, we significantly improve the results presented therein by considering exogenous disturbances acting on the controlled system and robustifying the proposed algorithm. Furthermore, we streamline the proofs and theoretical analysis and apply our framework in a numerical simulation to a tracking control problem for a mobile robot.

This paper is organized as follows. In Section \ref{sec:setting}, we discuss the setting considered in this work. The proposed algorithm is presented in Section \ref{sec:algorithm}, where we first design the reference governor and discuss our constraint tightening approach. Then, the overall framework is proposed. In Section \ref{sec:analysis}, we establish theoretical guarantees for the proposed algorithm on recursive feasibility, constraint satisfaction, a minimum rate of convergence, and a bound on its dynamic regret. Implementational aspects are illustrated in Section \ref{sec:simulation} on a numerical simulation of a tracking control problem for a mobile robot. Finally, Section \ref{sec:conclusion} concludes the paper.

\textit{Notation:} We denote the set of natural numbers (including $0$) by $\mathbb{N}$. For a vector $x\in\mathbb{R}$ and a matrix $A\in\mathbb{R}^{n\times m}$, $\norm{x}$ is the Euclidean norm and $\norm{A}$ is the corresponding induced matrix $2$-norm. The identity matrix of size $n\times n$ is denoted by $I_n$ and $0_{m\times n}\in\mathbb{R}^{m\times n}$ denotes the matrix of all zeros. For two sets $\mathcal S, \mathcal T\subseteq\mathbb{R}^n$, the Minkowski set sum is defined by $\mathcal S\oplus \mathcal T := \{s+t: s\in\mathcal{S},t\in\mathcal{T}\}$, the Pontryagin set difference is given by $\mathcal S\ominus\mathcal T := \{s\in\mathbb{R}^n: s+t\in\mathcal S,~\forall t\in\mathcal T\}$, and the interior of $\mathcal S$ is $\text{int }\mathcal S$. The projection operator of a point $x\in\mathbb{R}^n$ onto a set $\mathcal S\subseteq\mathbb R^{n}$ is defined by $\Pi_{\mathcal{S}}(x) := \arg\min_{s\in\mathcal S} \norm{x-s}^2$.

\section{SETTING} \label{sec:setting}

We consider a linear time-invariant system of the form

\begin{equation}
	x_{t+1} = Ax_t + Bu_t + B_w w_t \label{eq:original_dynamics}
\end{equation}
with some initial state $x_0\in\mathbb{R}^n$, where $x_t\in\mathbb{R}^n$ is the system state, $u_t\in\mathbb{R}^m$ is the system input and $w_t\in\mathbb R^o$ is a disturbance, all at time $t\in\mathbb{N}$. We assume that measurements of the full state $x_t$ are available. However, if only noisy state measurements are available due to, e.g., measurement noise or use of an observer, we show in Section~\ref{sec:simulation} that our results are still applicable using a suitable reformulation of the noisy closed-loop dynamics, compare also~\cite{Nonhoff24}. Furthermore, system \eqref{eq:original_dynamics} is subject to the constraints
\begin{equation}
	y_t \in \mathcal{Y}, \label{eq:constraints}
\end{equation}
where $y_t = C_ox_t + Du_t + D_w w_t$ is the constraint output, i.e., the goal is to keep the output $y_t$ in a constraint set $\mathcal Y\subseteq\mathbb R^p$ at all times $t\in\mathbb{N}$. Note that pure input or state constraints as well as mixed constraints can be considered by a suitable choice of the matrices $C_o$, $D$, and $D_w$. 

\begin{ass} \label{ass:stabilizable}
	The pair $(A,B)$ is stabilizable.
\end{ass}
\begin{ass} \label{ass:constraints}
	The constraint set $\mathcal Y$ is compact, convex, and $0\in\text{int }\mathcal Y$.
\end{ass}
\begin{ass} \label{ass:disturbance}
	There exists a compact, convex set $\mathcal{W}\subseteq\mathbb R^o$ such that $0\in\text{int }\mathcal{W}$ and $w_t\in\mathcal{W}$ for all times $t\in\mathbb{N}$.
\end{ass}

By Assumption \ref{ass:stabilizable}, we can design a state feedback matrix $K\in\mathbb{R}^{m\times n}$ such that $A_K:=A+BK$ is Schur stable. Thus, we can define a virtual reference input $v_t\in\mathbb{R}^m$ and set $u_t = v_t+Kx_t$ for all times $t\in\mathbb{N}$ to obtain the closed-loop dynamics and constraint output
\begin{subequations} \label{eq:system}
\begin{align}
	x_{t+1} &= A_Kx_t + Bv_t + B_ww_t \label{eq:dynamics} \\
	y_t &= Cx_t + Dv_t + D_ww_t, \label{eq:output}
\end{align}
\end{subequations}
where $C := C_o+DK$, subject to the constraints \eqref{eq:constraints}.

Additionally, we define the map from a virtual reference input $v_t$ to the corresponding steady state of the undisturbed stabilized system by $S_K := (I_n-A_K)^{-1}B$, i.e., $S_Kv_t = A_KS_Kv_t+Bv_t$ holds for any virtual reference $v_t\in\mathbb{R}^m$. The set of admissible steady-state references is given by $\mathcal S_v := \{v\in\mathbb{R}^m: (CS_K+D)v \in \mathcal Y\ominus D_w\mathcal W\}$. The goal is to optimize performance measured by time-varying, a priori unknown cost functions $L_t:\mathbb{R}^m\times\mathbb{R}^n\mapsto\mathbb{R}$. More specifically we aim to solve the optimization problem
\begin{equation}
	\min_{v_0,\dots,v_T\in\mathcal{S}_v} \sum_{t=0}^T L_t(v_t+Kx_t,x_t)\quad \text{s.t. \eqref{eq:constraints},\eqref{eq:system}}, \label{eq:OCP}
\end{equation}
where at each time $t$, only the previous cost functions $L_0,\dots,L_{t-1}$ are known. As discussed above, this setting frequently arises in control applications due to, e.g., time-varying and a priori unknown reference trajectories in tracking control problems or time-varying and a priori unknown exogenous inputs that influence the cost functions. However, we note that obtaining the optimal solution of~\eqref{eq:OCP} is generally intractable due to the time-varying and a priori unknown cost functions $L_t$. Instead, we develop an algorithm in Section~\ref{sec:algorithm} that aims to approximate the optimal performance by tracking the a priori unknown and time-varying optimal steady states using the OCO framework, while ensuring satisfaction of the constraints~\eqref{eq:constraints}. As common in OCO-based control \cite{Nonhoff24,Agarwal19,Lin2023,Li2019}, we assume some regularity of the time-varying cost functions $L_t$.

\begin{ass} \label{ass:cost_fcn}
	The cost functions $L_t(v,x)$ are $l_L$-Lipschitz continuous on the set $\mathcal Z:=\{(v,x)\in\mathbb{R}^m\times\mathbb{R}^n: Cx+Dv\in\mathcal Y\}$, i.e., there exists $l_L$ such that $\norm{L_t(v,x)-L_t(\tilde{v},\tilde{x})} \leq l_L \norm{(v,x)-(\tilde{v},\tilde{x})}$ holds for all $t\in\mathbb{N}$ and $(v,x),(\tilde{v},\tilde{x})\in\mathcal{Z}$. Moreover, the steady-state cost functions $L^s_t(v) := L_t(v+KS_Kv,S_Kv)$ are $\alpha_v$-strongly convex\footnote{Compare, e.g., \cite[Definition 2.1.3]{Nesterov18}.} and its gradients are $l_v$-Lipschitz continuous on $v\in\mathcal{S}_v$ for all $t\in\mathbb{N}$.
\end{ass}

Note that, if $\mathcal{Z}$ and $\mathcal{S}_v$ are compact, then local Lipschitz continuity of $L_t(v,x)$ and $\nabla L^s_t(v)$ implies Lipschitz continuity on $\mathcal{Z}$ and $\mathcal{S}_v$ as required in Assumption~\ref{ass:cost_fcn}.

\section{ROBUST OCO-RG-BASED CONTROL} \label{sec:algorithm}

In this work, we design an algorithm that separates the problem into minimizing the cost functions in \eqref{eq:OCP} and ensuring constraint satisfaction \eqref{eq:constraints} at all times $t\in\mathbb N$. For that, we apply the online gradient descent (OGD) algorithm \cite{Zinkevich03} to generate a sequence of references $r_t\in\mathcal{S}_v$ based on the time-varying cost functions $L_t$, and subsequently use a reference governor that computes a virtual reference $v_t\in\mathcal{S}_v$ based on the current reference $r_t$ and measured system state $x_t$. Note that, even though $r_t\in\mathcal{S}_v$ implies that the reference computed by the OGD algorithm is feasible at steady state, directly applying it to system \eqref{eq:system} may lead to constraint violation during the transient phase. In this case, the reference governor modifies the reference such that constraint satisfaction is guaranteed. This conceptually simple approach is illustrated in Figure~\ref{fig:BSB}.

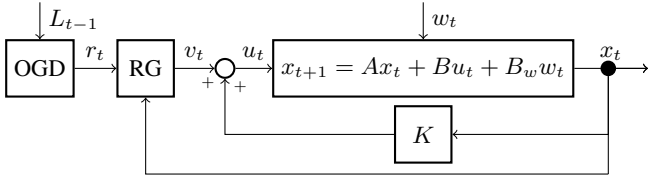
\begin{figure}	
	\small
	\begin{tikzpicture}[node distance = 40,scale = 1]
		\draw[draw=none] (-.5,-2.4) rectangle (8.1,0.3);
		\tikzstyle{block} = [draw, minimum size=.75cm, thick, align = center]
		\tikzstyle{add} = [draw, shape = circle, inner sep = 0pt, minimum size=.25cm, thick]
		\tikzstyle{circle} = [draw, shape=circle, inner sep = 0pt, minimum size = .15cm, fill=black, ultra thick]		
		
		\coordinate (start) at (0,0);	
		\node[block] (OCO) [below of = start,yshift=15] {OGD};		
		\node[block] (RG) [right of = OCO] {RG};
		\node[add] (add) [right of = RG,xshift=-10] {};
		\node[block] (system) [right of = add,xshift = 35] {$x_{t+1}=Ax_t+Bu_t+B_ww_t$};
		\node[circle] (circ) [right of = system,xshift=30] {};
		\node[block] (contr) [below of = system,yshift = 15] {$K$};
		\coordinate[right of = circ, xshift = -25] (end);
		\coordinate[below of = circ, yshift = -0] (aux);
		\coordinate[above of = system,yshift=-15] (dist);
		
		\draw[->] (start) -- (OCO) node[right,midway] {$L_{t-1}$};
		\draw[->] (OCO) -- (RG) node[above,midway] {$r_t$};
		\draw[->] (RG) -- (add) node[above,midway] {$v_t$} node[below,midway,xshift=5] {\tiny{$+$}};
		\draw[->] (add) -- (system) node[above,midway] {$u_t$};
		\draw[->] (system) -- (end) node[above,midway] {$x_t$};
		\draw[->] (circ) |- (contr);
		\draw[->] (contr) -| (add) node[right,yshift=-7] {\tiny{$+$}};
		\draw (circ) -- (aux);	
		\draw[->] (aux) -| (RG);
		\draw[->] (circ)--(end);
		\draw[->] (dist)--(system) node[right,midway] {$w_t$};
	\end{tikzpicture}	
	\caption{Block diagram of the robust OCO-RG scheme.}
	\label{fig:BSB}
\end{figure}

\subsection{Design of the reference governor} \label{subsec:RG}

In this section, we develop a reference governor that ensures robust satisfaction of the constraints \eqref{eq:constraints}. The proposed approach robustifies the results in \cite{Nonhoff22b} by taking the disturbances $w_t$ into account. As in \cite{Nonhoff22b}, the proposed reference governor is based on the results presented in \cite{Kalabic14}. 

As described above, reference governors compute a feasible virtual reference signal $v_t$ at each time $t$ that ensures constraint satisfaction. More specifically, the virtual reference $v_t$ is chosen such that, if it were applied constantly to the system \eqref{eq:system}, then the constraints \eqref{eq:constraints} would be satisfied for all future time steps. Such a virtual reference is computed using the maximal output admissible set (MAS) \cite{Gilbert91}.
\begin{defn} \label{def:MAS}
	The maximal output admissible set $\MAS$ of a system $\chi_{t+1} = A\chi_t+B_w \omega_t$, $\omega_t\in\mathcal W$, subject to constraints $y_t=C\chi_t+D_w\omega_t\in\mathcal Y$ is defined as
	\begin{align*}
		\MAS := \{x\in\mathbb{R}^n: \chi_0 = x,~\chi_{t+1} = A\chi_t+B_w\omega_t,& \\
		 C\chi_t+D_w\omega_t\in\mathcal Y\quad\forall \omega_t\in\mathcal{W},t\in\mathbb{N}\}&
	\end{align*}
\end{defn}
The MAS is the set of all initial conditions for which constraint satisfaction is guaranteed for any sequence of disturbances $w_t\in\mathcal W$ and for all times. Furthermore, in order to ensure a sufficient rate of convergence, we employ a $\lambda$-contractive set in the proposed reference governor scheme.
\begin{defn} \label{def:contr_set}
	For a system $\chi_{t+1} = A\chi_t+B_w\omega_t$, $\omega_t\in\mathcal W$, a set $\mathcal{X}\in\mathbb{R}^n$ is $\lambda$-contractive with some $\lambda\in(0,1)$ if $\mathcal{X}$ is convex, closed, $0\in\text{int }\mathcal{X}$, and $A\mathcal X\oplus B_w\mathcal W \subseteq \lambda \mathcal X$.
\end{defn}

Finally, recall that $S_K = (I_n-A_K)^{-1}B$ maps from a virtual reference $v_t$ to the corresponding steady state of the undisturbed stabilized system. We define the error between the state $x_t$ and the steady state corresponding to the reference $v_t$ by
\[
	e_t := x_t - S_K v_t.
\]
Next, consider a constant reference $v_t\equiv\nu$. Then, the system dynamics \eqref{eq:system} for an extended state written in error coordinates are given by
\begin{subequations} \label{eq:error_system}
\begin{align}
	v_{t+1} &= v_t \\
	e_{t+1} &= A_K e_t + B_w w_t \label{eq:error_dynamics} \\
	y_t &= Ce_t + (CS_K+D)v_t + D_w w_t,
\end{align}
\end{subequations}
with $v_0 = \nu$ and $e_0 = x_0 - S_k\nu$. Similar to \cite{Nonhoff22b,Kalabic14}, we denote by $\MAS^\lambda$ the MAS of system \eqref{eq:error_system} with \eqref{eq:error_dynamics} replaced by
\begin{equation}
	\epsilon_{t+1} = \frac{1}{\lambda} \left(A_K \epsilon_t + B_w w_t \right) \label{eq:contractive_error_dynamics}
\end{equation}
for some $\rho(A_K)<\lambda<1$ and with constraints $y_t\in\mathcal{Y}$. We make use of this MAS $\MAS^\lambda$ in our OCO-RG scheme in Section \ref{subsec:oco-rg-scheme}. Note that, since $\mathcal Y$ is compact, convex, and $0\in\text{int }\mathcal Y$ by Assumption \ref{ass:constraints}, the MAS $\MAS^\lambda$ is closed, convex, and $0\in\text{int }\MAS^\lambda$ \cite{Gilbert91}.

\subsection{Constraint tightening}
In this section, we develop a constraint tightening that ensures robustness of the proposed approach with respect to the disturbances $w_t$ acting on system \eqref{eq:system}. Similar to robust model predictive control approaches in the literature \cite{Rawlings2017}, we rely on constraint tightening based on a robust positively invariant (RPI) set in order to ensure robust constraint satisfaction for the proposed OCO-RG scheme.
\begin{defn} \label{def:RPI}
	A set $\mathcal{P}$ is called robust positively invariant for a system $\chi_{t+1} = A\chi+B_w\omega_t$, $\omega_t\in\mathcal W$, if $A\mathcal{P}\oplus B_w\mathcal W \subseteq \mathcal{P}$ holds.
\end{defn}
Let $\RPI$ be an RPI set for system \eqref{eq:contractive_error_dynamics}. It can be shown that the minimal RPI set $\RPI_\infty$ for system \eqref{eq:contractive_error_dynamics} is given by
\begin{equation}
	\RPI_\infty := \bigoplus_{i=0}^\infty \left( \left(\frac{1}{\lambda} A_K\right)^i \frac{1}{\lambda} B_w\mathcal{W} \right) \subseteq \RPI, \label{eq:mRPI}
\end{equation}
compare \cite{Gilbert91}. Let $\overbar{\mathcal{Y}}$ be a convex, compact subset of $\mathcal{Y}$ such that $0\in\overbar{\mathcal{Y}}\subseteq\text{int }\mathcal{Y}$. Then, we define the tightened set of admissible steady-state references by
\[
	\overbar{\mathcal{S}}_v := \{v\in\mathbb{R}^m:(CS_K+D)v\in\overbar{\mathcal{Y}}\ominus D_w\mathcal{W}\ominus C\RPI\}.
\]
Note that $\overbar{\mathcal{S}}_v$ is convex, because $\overbar{\mathcal{Y}}$ is convex. In the following, we use the tightened set $\overbar{\mathcal{S}}_v$ within the OGD algorithm.


\subsection{Robust OCO-RG scheme} \label{subsec:oco-rg-scheme}

\begin{figure} 
	\vspace{5pt}
	{
		\fbox{\parbox{.94\linewidth}{
				{\underline{Algorithm 1: Robust OCO-RG scheme}}
				
				\vspace{7pt}
				
				\noindent At $t=0$: Set $\alpha_0 = 1$, $v_0=r_0$ and apply $u_0 = v_0 + Kx_0$.
				
				\noindent At each time $t\in\mathbb N_{\geq 1}$, given the current system state~$x_t$:
			
				\noindent [S1] Online gradient descent
				\vspace{-1ex}
				\begin{equation}
					r_t = \Pi_{\overbar{S}_v} \left( r_{t-1} - \gamma\nabla L^s_{t-1}(r_{t-1}) \right), \label{eq:OCO_OGD}
				\end{equation}
				
				\noindent [S2] Reference governor
				\begin{subequations} \label{eq:RG}
					\begin{align}
						\begin{split}
						&\alpha_t = \max_{\alpha \in [0,1]} \alpha \quad \text{s.t. } (v,x_t-S_Kv) \in \MAS^\lambda \label{eq:RG_opt} \\
						&\phantom{\alpha_t = \max_{\alpha \in [0,1]} \alpha \quad} v = v_{t-1} + \alpha (r_t - v_{t-1}) 
						\end{split} \\
						&v_t = v_{t-1} + \alpha_t (r_t - v_{t-1}), \label{eq:RG_update}
					\end{align}
				\end{subequations}
				
				\noindent [S3] Control input
				\vspace{-1ex}
				\begin{equation}
					u_t = v_t + Kx_t. \label{eq:algo_input}
				\end{equation}
				\vspace{-15pt}
		}}
	}
\end{figure}

In this section, we introduce the proposed combination of OGD and the reference governor. The robust OCO-RG scheme is given in Algorithm~1. Compared to \cite{Nonhoff22b}, the main differences are in the definitions of the MAS $\MAS^\lambda$ and the set $\overbar{\mathcal{S}}_v$, which were modified in order to take disturbances into account and ensure robust constraint satisfaction.

As described above, the cost functions $L_t$ are time-varying and a priori unknown, i.e., at each time $t\in\mathbb{N}$, we only have access to the previous cost functions $L_0,\dots,L_{t-1}$. Therefore, it is generally impossible to compute the minimizing input of \eqref{eq:OCP} online. Instead, similar to, e.g., \cite{Nonhoff22b,Nonhoff24}, and the feedback optimization framework (compare Section~\ref{sec:intro}), Algorithm 1 aims to track the a priori unknown and time-varying optimal steady states $(\eta_t,\theta_t)$ given by
\begin{equation*}
	\eta_t := \arg\min_{r\in\overbar{\mathcal{S}}_v} L^s_t(r), \qquad
	\theta_t := S_K\eta_t.
\end{equation*}
Moreover, at time $t=0$, no cost functions are known to Algorithm~1 and it simply applies an initial reference $r_0$. In order to ensure constraint satisfaction at all times, we assume that the initial reference $r_0$ is feasible.

\begin{ass} \label{ass:init}
	The initial reference $r_0\in\overbar{\mathcal{S}}_v$ satisfies $(r_0,x_0-S_Kr_0)\in\MAS^\lambda$.
\end{ass}

Thus, if Assumption~\ref{ass:init} is satisfied, then $u_0=v_0+Kx_0$ with $v_0=r_0$ is a feasible input at time $t=0$. At each of the following time steps $t\in\mathbb{N}_{\geq1}$, the algorithm first measures the current system state $x_t$, and then
\begin{enumerate}
	\item[{[S1]}] applies one OGD step with a suitable step size $\gamma\in(0,\frac{2}{\alpha_v+l_v}]$ to the time-varying optimization problem $\{\min_r L^s_{t-1}(r)$ s.t. $r\in\overbar{\mathcal{S}}_v\}$ in order to compute an estimate of the optimal reference $\eta_{t-1}$. Note that we apply the tightened set $\overbar{\mathcal{S}}_v$ in order to obtain references that robustly satisfy the constraints at steady state.
	\item[{[S2]}] employs a reference governor to guarantee constraint satisfaction during the transient phases. The constraint \eqref{eq:RG_opt} ensures that the error $e_t=x_t-S_Kv_t$ lies in the MAS $\MAS^\lambda$ derived in Section~\ref{subsec:RG}, which implies robust constraint satisfaction as shown in Lemma~\ref{lem:rec_feas} below. The scalar optimization problem \eqref{eq:RG_opt} can be solved efficiently by, e.g., bisection.
	\item[{[S3]}] applies the control input $u_t = v_t+Kx_t$ to system \eqref{eq:original_dynamics}.
\end{enumerate}
Then, the cost function $L_t$ is received and [S1] to [S3] are repeated at the next time step $t+1$.

\section{THEORETICAL ANALYSIS} \label{sec:analysis}
In this section, we derive theoretical guarantees for Algorithm~1. Namely, we show that the reference governor \eqref{eq:RG} is recursively feasible and that the dynamic regret of the closed loop emerging from application of Algorithm~1 is bounded.
\begin{defn}
	Algorithm~1 is called recursively feasible, if existence of a solution to \eqref{eq:RG_opt} at $t=0$ implies that there exists a feasible solution to \eqref{eq:RG_opt} for all $t\in\mathbb{N}$.
\end{defn}
Recursive feasibility is crucial to guarantee that Algorithm~1 is well-defined at all times.
\begin{lem} \label{lem:rec_feas}
	Suppose Assumptions~\ref{ass:stabilizable}--\ref{ass:init} are satisfied. Then, Algorithm~1 is recursively feasible. Moreover, the constraints \eqref{eq:constraints} are satisfied for all $t\in\mathbb{N}$.
\end{lem}
\begin{proof}
	In this proof, we extend the arguments from the proofs of \cite[Lemmas~6 and~8]{Nonhoff22b} to account for the disturbances in \eqref{eq:original_dynamics}. We prove recursive feasibility by showing that $\alpha_t=0$ is a feasible solution to \eqref{eq:RG_opt} for all $t\in\mathbb{N}$. To this end, we need three auxiliary results. Let $\mathcal{E}^\lambda_\infty(v) := \{e\in\mathbb{R}^n: (v,e)\in\MAS^\lambda\}$. 
	
	First, we show that $\mathcal{E}^\lambda_\infty(v)$ is $\lambda$-contractive for~\eqref{eq:error_dynamics}, i.e., $e_t\in\mathcal{E}^\lambda_\infty(v)$ implies $A_Ke_t+B_ww_t\in\lambda\mathcal{E}^\lambda_\infty(v)$ for any $w_t\in\mathcal{W}$. Fix any $t\in\mathbb{N}$, $v\in\overbar{\mathcal{S}}_v$, and $w_t\in\mathcal{W}$. Using the fact that the MAS $\MAS^\lambda$ is computed for the system \eqref{eq:error_system} with \eqref{eq:error_dynamics} replaced by \eqref{eq:contractive_error_dynamics}, and that the MAS is robust positively invariant by definition, we have
	\begin{align}
		e_t \in\mathcal{E}^\lambda_\infty(v) 
		&\refrarrow{\eqref{eq:contractive_error_dynamics}} \left( v,\frac{1}{\lambda}(A_Ke_t+B_ww_t)\right) \in \MAS^\lambda \nonumber \\
		&\reflrarrow {\eqref{eq:error_dynamics}} e_{t+1} = A_Ke_t+B_ww_t\in\lambda\mathcal{E}^\lambda_\infty(v). \label{eq:lambda_contractive_set}
	\end{align}
	
	Second, we show that $0\in\mathcal{E}^\lambda_\infty(v)$ for any $v\in\overbar{\mathcal{S}}_v$. To see this, note that for $\epsilon_0=0$ and any $t\in\mathbb{N}$, $0\in\RPI$ implies \mbox{$	\epsilon_t \refin{\eqref{eq:contractive_error_dynamics}} \bigoplus_{i=0}^{t-1} \left( \frac{1}{\lambda^{i+1}}A_K^i B_w \mathcal W\right) \refsubseteq{\eqref{eq:mRPI}} \RPI_\infty \subseteq \RPI$.} Hence, for $\epsilon_0=0$, any $t\in\mathbb{N}$, and any $v\in\overbar{\mathcal{S}}_v$, we get $(CS_K+D)v+C\epsilon_t + D_ww_t\in\overbar{\mathcal{Y}}\ominus D_w\mathcal{W} \ominus C\RPI \oplus C\RPI \oplus D_w\mathcal{W} \subseteq \overbar{\mathcal{Y}}$, which implies the desired result $0\in\mathcal{E}^\lambda_\infty(v)$. Combining this with \eqref{eq:lambda_contractive_set} and using convexity of $\mathcal{E}^\lambda_\infty(v)$ (which follows from convexity of $\MAS^\lambda$), we get
	\begin{equation}
		e_t \in \mathcal{E}^\lambda_\infty(v) \refrarrow{\eqref{eq:lambda_contractive_set}} e_{t+1}\in\lambda\mathcal{E}^\lambda_\infty(v) \subseteq \mathcal{E}^\lambda_\infty(v). \label{eq:pos_inv_set}
	\end{equation}
	
	Third, note that $r_t\in\overbar{\mathcal{S}}_v$ for all $t\in\mathbb{N}$ by definition in \eqref{eq:OCO_OGD} and Assumption~\ref{ass:init}. We show that also $v_t\in\overbar{\mathcal{S}}_v$ for all $t\in\mathbb{N}$ by induction. At $t=0$, $v_0\in\overbar{\mathcal{S}}_v$ holds by Assumption~\ref{ass:init}. Assume that $v_{t-1}\in\overbar{\mathcal{S}}_v$. Due to convexity of $\overbar{\mathcal{S}}_v$, $r_t\in\overbar{\mathcal{S}}_v$, and $v_{t}$ being a convex combination of $r_t$ and $v_{t-1}$ by \eqref{eq:RG_update}, we have $v_t\in\overbar{\mathcal{S}}_v$, which concludes the proof by induction.
	
	Finally, we show recursive feasibility of Algorithm~1, i.e., we prove that there exists a feasible solution to \eqref{eq:RG_opt} for all $t\in\mathbb{N}$. Again, we show the result by induction. For $t=0$, $(v_0,e_0)\in\MAS^\lambda$ by Assumption~\ref{ass:init}. Fix any $t\in\mathbb{N}_{\geq1}$ and assume that $(v_{t-1},e_{t-1})\in\MAS^\lambda$. Consider the candidate solution $\alpha^c_t=0$. It remains to show that $v^c_t=v_{t-1}+\alpha^c_t(r_t-v_{t-1})=v_{t-1}$ satisfies $(v^c_t,e_{t})\in\MAS^\lambda$. Since $(v_{t-1},e_{t-1})\in\MAS^\lambda \Leftrightarrow e_{t-1}\in\mathcal{E}^\lambda_\infty(v^c_t)$, implies $(v^c_t,e_{t})\in\MAS^\lambda$ by \eqref{eq:pos_inv_set}, Algorithm 1 is recursively feasible. Finally, note that $y_t\in\mathcal{Y}$ by definition of the MAS $\MAS^\lambda$.
\end{proof}

Having established recursive feasibility and robust constraint satisfaction for all times $t\in\mathbb{N}$, we analyze the closed-loop performance of Algorithm~1. As in previous works, e.g., \cite{Nonhoff24,Nonhoff22b,Li2019}, we use dynamic regret as a measure of performance. We define dynamic regret as the accumulated difference between the closed-loop cost and the optimal steady-state cost
\[
	\regret:=\sum_{t=0}^T L_t(u_t,x_t) - L_t(\eta_t,\theta_t).
\]
In order to obtain an upper bound on the dynamic regret, we first prove that the modified reference $v_t$ is moved closer to $r_t$ at all times $t\in\mathbb{N}$.
\begin{lem}\label{lem:convergence}
	Suppose Assumptions~\ref{ass:stabilizable}--\ref{ass:init} are satisfied. There exists $\delta\in(0,1]$ such that $\alpha_t\geq\delta$ for all $t\in\mathbb{N}$.
\end{lem}
\begin{proof}
	This proof follows the same steps as the proof of \cite[Lemma 9]{Nonhoff22b}. Fix any $t\in\mathbb{N}_{\geq1}$. Assumption~\ref{ass:init} and Lemma~\ref{lem:rec_feas} imply $(v_{t-1},x_{t-1}-S_Kv_{t-1})\in\MAS^\lambda$ and, thus, 
	\begin{align}
		(v_{t-1},&e_{t-1})\in\MAS^\lambda \refrarrow{\eqref{eq:lambda_contractive_set}} A_Ke_{t-1}+B_ww_{t-1}\in\lambda\mathcal{E}^\lambda_\infty(v_{t-1}) \nonumber \\
		&\reflrarrow{\eqref{eq:dynamics}} \left( v_{t-1}, \frac{1}{\lambda}(x_t-S_Kv_{t-1}) \right) \in \MAS^\lambda. \label{eq:disturbances_cancel_out}
	\end{align}
	This is the same result as \cite[Eq. (13)]{Nonhoff22b} (in particular, note that the left hand side of \eqref{eq:disturbances_cancel_out} does not explicitly depend on $w_{t-1}$) and the remainder of the proof follows exactly the same steps.
\end{proof}

Using Lemma~\ref{lem:convergence}, we can now establish an upper bound on the dynamic regret of Algorithm 1.
\begin{thm} \label{thm}
	Suppose Assumptions~\ref{ass:stabilizable}--\ref{ass:init} are satisfied and $\gamma\in(0,\frac{2}{\alpha_v+l_v}]$. There exist $c_0,c_\eta,c_w>0$ such that
	\begin{equation}
		\regret \leq c_0+c_\eta\sum_{t=1}^T \norm{\eta_t-\eta_{t-1}} + c_w\sum_{t=0}^{T-1} \norm{w_t} \label{eq:regret_bound}
	\end{equation}
	holds for any $T\in\mathbb{N}$, any sequence of disturbances $w_t$, and any sequence of cost functions $L_t$.
\end{thm}
\begin{proof}
	This proof follows the same steps as the proof of \cite[Theorem 10]{Nonhoff22b}, suitably adapted to account for the disturbances. First, Assumption~\ref{ass:cost_fcn} yields
	\begin{align}
		&\regret = \sum_{t=0}^T L_t(v_t+Kx_t,x_t) - L_t(\eta_t+KS_K\eta_t,S_K\eta_t) \nonumber \\
		\leq &l_L \sum_{t=0}^T \norm{\begin{bmatrix} v_t+Kx_t \\ x_t \end{bmatrix} - \begin{bmatrix} \eta_t+KS_K\eta_t \\ S_K\eta_t\end{bmatrix} } \nonumber \\
		\leq &l_L \sum_{t=0}^T \norm{ \begin{bmatrix} Kx_t - KS_Kv_t \\ x_t - S_K v_t \end{bmatrix} } + \norm{\begin{bmatrix} (I_m+KS_K) (v_t-\eta_t) \\ S_Kv_t-S_K\eta_t \end{bmatrix} } \nonumber \\
		\leq &c_x \sum_{t=0}^T \norm{x_t-S_Kv_t} + c_v \sum_{t=0}^T \norm{v_t-\eta_t}, \label{eq:regret_proof_1}
	\end{align}
	where $c_v:=l_L\norm{\begin{bmatrix} (I_m+KS_K)^\top & S_K^\top \end{bmatrix}^\top}$ and $c_x:=l_L \norm{\begin{bmatrix} K^\top & I_n \end{bmatrix}^\top}$. To bound the two sums in \eqref{eq:regret_proof_1}, we use the following standard result from convex optimization.
	\begin{lem} \cite[Theorem 2.2.14]{Nesterov18}
		Suppose Assumptions \ref{ass:constraints} and \ref{ass:cost_fcn} are satisfied and let $\gamma\in(0,\frac{2}{\alpha_v+l_v}]$. Then,
		\begin{equation}
			\norm{\Pi_{\overbar{\mathcal{S}}_v} (r-\gamma\nabla L^s_{t-1}(r))-\eta_{t-1}} \leq \kappa\norm{r-\eta_{t-1}}, \label{eq:GD_convergence}
		\end{equation}
		where $\kappa = 1-\gamma\alpha_v\in[0,1)$, holds for any $r\in\overbar{\mathcal{S}}_v$, cost function $L^s_{t-1}$, and any $t\in\mathbb{N}_{\geq1}$.
	\end{lem}
	First, we bound the second sum in \eqref{eq:regret_proof_1}. Note that
	\begin{align*}
		&\sum_{t=1}^T \norm{r_t-\eta_{t-1}} \refleq{\eqref{eq:GD_convergence}} \kappa\sum_{t=1}^T \norm{r_{t-1}-\eta_{t-1}} \\
		&\leq \kappa\norm{r_0-\eta_0} + \kappa \sum_{t=1}^T \norm{r_t-\eta_{t-1}} + \kappa \sum_{t=1}^{T} \norm{\eta_t-\eta_{t-1}}
	\end{align*}
	implies
	\begin{equation}
		\sum_{t=1}^T \norm{r_t{-}\eta_{t-1}} \leq \frac{\kappa}{1-\kappa} \norm{r_0{-}\eta_0} + \frac{\kappa}{1-\kappa} \sum_{t=1}^{T} \norm{\eta_t{-}\eta_{t-1}}. \label{eq:regret_proof_aux_1}
	\end{equation}
	Together with \eqref{eq:RG_update} and $\alpha_t\in[\delta,1]$ in Lemma~\ref{lem:convergence}, this implies
	\begin{align*}		
		\sum_{t=0}^T \norm{v_t-\eta_t} &\refleq{\eqref{eq:RG_update}} \norm{v_0-\eta_0} + \sum_{t=1}^T (1-\alpha_t) \norm{v_{t-1}-\eta_{t-1}} \\
		&\quad+ \sum_{t=1}^T \alpha_t \norm{r_t-\eta_{t-1}} + \sum_{t=1}^T \norm{\eta_t-\eta_{t-1}} \\
		&\refleq{\eqref{eq:regret_proof_aux_1}} k_0\norm{v_0-\eta_0} + (1-\delta) \sum_{t=1}^{T} \norm{v_t-\eta_t} \\
		&\quad+ \frac{1}{1-\kappa} \sum_{t=1}^T \norm{\eta_t-\eta_{t-1}},
	\end{align*}
	where $k_0:= \frac{1+(1-\kappa)(1-\delta)}{1-\kappa}$. From here, rearranging yields the desired bound
	\begin{equation}
		\sum_{t=0}^T \norm{v_t-\eta_t} \leq \frac{k_0}{\delta}\norm{v_0-\eta_0} + \frac{1}{\delta(1-\kappa)} \sum_{t=0}^T \norm{\eta_t-\eta_{t-1}}. \label{eq:regret_proof_sum_2}
	\end{equation}
	Note that \eqref{eq:regret_proof_sum_2} together with \eqref{eq:regret_proof_aux_1} and $\alpha_t\leq1$ implies
	\begin{align}
		&\sum_{t=1}^T \norm{v_t-v_{t-1}} \refeq{\eqref{eq:RG_update}} \sum_{t=1}^T \alpha_t \norm{r_t-v_{t-1}} \nonumber \\
		&\leq \norm{v_0-\eta_0} + \sum_{t=1}^T \norm{r_t-\eta_{t-1}} + \sum_{t=1}^T \norm{v_t-\eta_t} \nonumber \\
		&\refleq{\eqref{eq:regret_proof_aux_1},\eqref{eq:regret_proof_sum_2}} c_{0,v} \norm{v_0-\eta_0} + \frac{1+\delta\kappa}{\delta(1-\kappa)} \sum_{t=1}^T \norm{\eta_t-\eta_{t-1}}, \label{eq:regret_proof_aux_2}
	\end{align}
	where $c_{0,v} = \frac{k_0(1-\kappa)+\delta}{\delta(1-\kappa)}$. It remains to bound the first sum in \eqref{eq:regret_proof_1}. To do that, consider the error dynamics of $e_t=x_t-S_Kv_t$ given by $e_{t+1}=A_Ke_t + S_K(v_t-v_{t+1}) + B_ww_t$. Repeatedly using this equation yields
	\begin{equation}
		e_t = A_K^t e_0 + \sum_{i=0}^{t-1} A_K^i S_K (v_{t-i-1}-v_{t-i}) + \sum_{i=0}^{t-1} A_K^i B_w w_{t-i-1}. \label{eq:closed_loop_error_dynamics}
	\end{equation}
	Furthermore, since $A_K$ is Schur stable, there exists $c_A\geq1$ and $\sigma\in(0,1)$ such that $\norm{A_K^t}\leq c_A\sigma^t$ holds for any $t\in\mathbb{N}$. Let $c_\sigma:=c_A(1-\sigma)^{-1}$. Then, we get
	\begin{align}
		&\sum_{t=0}^T \norm{x_t-S_Kv_t} \refleq{\eqref{eq:closed_loop_error_dynamics}} \norm{e_0}+ \sum_{t=1}^T \norm{A_K^t}\norm{e_0} \nonumber \\
		&\qquad+ \norm{S_K}\sum_{t=1}^T \sum_{i=0}^{t-1} \norm{A_K^i} \norm{v_{t-i-1}-v_{t-i}}\nonumber \\
		&\qquad+\norm{B_w} \sum_{t=1}^T \sum_{i=0}^{t-1} \norm{A_K^i} \norm{w_{t-i-1}}\nonumber \\
		&\leq c_\sigma \norm{e_0}+\norm{S_K} \sum_{t=1}^{T} \left( \norm{v_t-v_{t-1}} \sum_{i=0}^{T-1} \norm{A_K^i} \right)\nonumber \\
		&\qquad + \norm{B_w} \sum_{t=0}^{T-1} \left( \norm{w_t} \sum_{i=0}^{T-1} \norm{A_K^i} \right)\nonumber \\
		&\leq c_\sigma \norm{e_0}+c_\sigma\norm{S_K} \sum_{t=1}^T \norm{v_t{-}v_{t-1}} + c_\sigma \norm{B_w} \sum_{t=0}^{T-1} \norm{w_t}. \label{eq:regret_proof_sum_1}
	\end{align}
	From here, plugging \eqref{eq:regret_proof_aux_2} into \eqref{eq:regret_proof_sum_1} and the result together with \eqref{eq:regret_proof_sum_2} into \eqref{eq:regret_proof_1} yields the desired bound.
\end{proof}

Theorem \ref{thm} establishes an upper bound \eqref{eq:regret_bound} on the dynamic regret of Algorithm~1 that depends linearly on the quantity $\sum_{t=1}^T \norm{\eta_t-\eta_{t-1}}$, which can be seen as a measure of the variation of the cost functions and is commonly termed path length in the literature. Compared to \cite[Theorem 10]{Nonhoff22b}, we obtain an additional term that depends linearly on the magnitude of the disturbances $\sum_{t=0}^{T-1} \norm{w_t}$. This has to be expected, because the disturbance may drive the closed loop away from the optimal steady state. Note that a similar bound depending on the magnitude of the disturbances has recently been established for a different method \cite{Nonhoff24} and we refer the reader to \cite[Section~IV]{Nonhoff24} for a more detailed discussion.

\section{NUMERICAL EXAMPLE} \label{sec:simulation}

\setlength{\breite}{.4\textwidth}
\setlength{\hohe}{5cm}
\begin{figure}
	\small
	\input{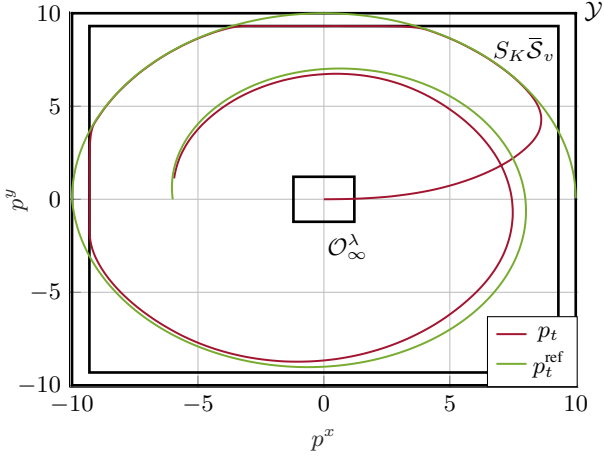}
	\caption{Closed-loop position $p_t$ (red) and reference position $p_t^{\text{ref}}$ (green) together with the projections of the sets $\mathcal{Y}$, $S_K\overbar{\mathcal{S}}_v$, and $\MAS^\lambda$ (centered at $0$) onto the $(p^x,p^y)$-plane (black)}
	\label{fig:position}
\end{figure}

In this section, we implement the proposed robust OCO-RG scheme on a robot moving on a two-dimensional plane\footnote{The code for this simulation can be found online at \url{https://doi.org/10.25835/x2q5hrrg}}. The dynamics of the robot are given by
\[
	\underbrace{\begin{bmatrix} p_{t+1} \\ \nu_{t+1} \end{bmatrix}}_{x_{t+1}} = \underbrace{\begin{bmatrix}I_2&\tau I_2\\0_{2\times2}&I_2\end{bmatrix}}_{A} \underbrace{\begin{bmatrix} p_t \\ \nu_t \end{bmatrix}}_{x_t} + \underbrace{\begin{bmatrix}0_{2\times2}\\\tau I_2\end{bmatrix}}_Bu_t,
\]
where $\tau=0.1$, $p_t= (p^x_t,p^y_t)$ denotes the robot's location, $\nu_t=(\nu^x_t,\nu^y_t)$ its velocity, and $u_t\in\mathbb{R}^2$ is the control input, and $x_0 = 0_{4\times1}$. Furthermore, only noisy measurements of the robot's state are available $\tilde x_t := x_t + \mu_t$, where $\mu_t\in\mathcal{M} := \{\mu\in\mathbb{R}^n: \norm{\mu}_\infty\leq0.01\}$ denotes the measurement error, caused by, e.g., measurement noise or use of an observer. While the results derived in this paper assumed perfect state measurements, the following reformulation shows that also this setting including measurement noise can be handled by our proposed algorithm: First, we design the stabilizing feedback $K$ as the linear-quadratic regulator (LQR) with weighting matrices $Q=100I_n$ and $R=1$. Then, the noisy closed-loop dynamics are given by
\[
	\tilde x_{t+1} = A_K\tilde x_t+Bv_t+\mu_{t+1}-A_K\mu_t.
\]
Finally, the constraints are $\norm{p_t}_\infty\leq10$, $\norm{\nu_t}_\infty\leq1$, and $\norm{u_t}_\infty\leq2$. Therefore, we set
\[
	w_t = \begin{bmatrix} \mu_{t+1}-A_K\mu_t\\ \mu_t\end{bmatrix},~\mathcal{W} = \begin{bmatrix} I_n \\ 0_{n\times n}\end{bmatrix} \mathcal{M} \oplus \begin{bmatrix} -A_K \\ I_n \end{bmatrix} \mathcal{M},
\]
$B_w = \begin{bmatrix} I_n & 0_{n\times n}\end{bmatrix}$, and for the constraints
\[
	C_0=\begin{bmatrix} I_n \\ 0_{m\times n} \end{bmatrix},~D=\begin{bmatrix} 0_{n\times m} \\ I_m \end{bmatrix},~D_w =\begin{bmatrix} 0_{n\times n} & -I_n \\ 0_{m\times n} & 0_{m\times n} \end{bmatrix},	
\]
and $\mathcal{Y} = \{y\in\mathbb{R}^6: \norm{y_{1:2}}_\infty\leq10,~\norm{y_{3:4}}_\infty\leq1,\norm{y_{5:6}}_\infty\leq2\}$, where $y_{i:j}$ denotes the vector consisting of the $i$-th to $j$-th entry of the vector $y$. Since the system and constraints are now in the form \eqref{eq:system}, \eqref{eq:constraints}, we proceed to design the OGD algorithm and the reference governor. For the reference governor, we set $\lambda = 0.99$ and compute the MAS $\MAS^\lambda$ using the mpt3 toolbox \cite{MPT3} and the method described in \cite{Kalabic14,Gilbert91}. The optimization problem \eqref{eq:RG_opt} is solved by bisection. For the OGD algorithm, we set $\gamma = 17.35$ (which is such that the condition in Theorem~\ref{thm} holds), $\overbar{\mathcal{Y}}=0.95\mathcal{Y}$, and compute the RPI set with the method from \cite{Rakovic2005}. The initial reference is $r_0=0_{2\times1}$. We aim to follow an a priori unknown and time-varying reference for the position $p_t$, i.e., $L_t(u,x) = L_t(p) = \frac{1}{2}\norm{p-p^{\text{ref}}_t}^2$. The reference position $p^{\text{ref}}_t$ is depicted in Figure~\ref{fig:position} together with the closed-loop position $p_t$ resulting from application of Algorithm~1, and the projection of the sets $\mathcal Y$, $S_K\overbar{\mathcal{S}}_v$, and $\MAS^\lambda$ onto the $(p^x,p^y)$-plane. For the first $600$ time steps, the reference position $p^{\text{ref}}_t$ moves along a circle with radius $10$ from $p^{\text{ref}}_0=(10,0)$ to $p^{\text{ref}}_{600}=(-10,0)$. Afterwards, the reference accelerates and the radius decreases both linearly in time. It can be seen that initially, after a transient phase, the closed loop tracks the reference position closely while the reference is inside the projection of $S_K\overbar{\mathcal{S}}_v$ onto the $(p^x,p^y)$-plane and moves to the border of the tightened constraints otherwise. When the reference accelerates, the closed loop is not able to keep up anymore due to the constraints. This can also be seen in Figure~\ref{fig:alpha}, where the reference governor intervenes when the reference accelerates, i.e., after $t=600$, $\alpha_t<1$. However, we observe that $\alpha_t\geq0.005$ for all $t$ as predicted in Lemma~\ref{lem:convergence}. Moreover, Figure~\ref{fig:velocity} shows the closed-loop velocity $\nu^x_t$ together with the constraints. The reference governor successfully satisfies the constraints despite the presence of measurement noise, i.e., $\max_t \norm{\nu^x_t} = 0.9712<1$.

\setlength{\hohe}{4cm}
\begin{figure}
	\small
	\input{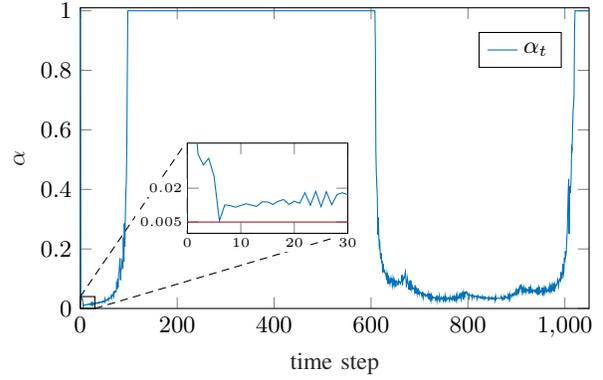}
	\caption{Reference governor parameter $\alpha_t$ (blue; compare \eqref{eq:RG_opt})}
	\label{fig:alpha}
\end{figure}

\begin{figure}
	\small
	\input{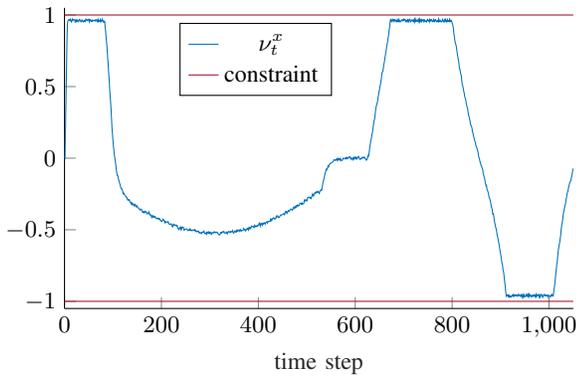}
	\caption{Closed-loop velocity $\nu^x_t$ (in direction of $p^x$; blue) together with the corresponding velocity constraints (red)}
	\label{fig:velocity}
\end{figure}

\section{CONCLUSION} \label{sec:conclusion}

In this paper, we propose a conceptually simple approach for controlling linear dynamical systems that allows to consider (i) time-varying and a priori unknown cost functions, (ii) state and input constraints, and (iii) exogenous disturbances, and solve each of these problems separately by suitably combining (i) the OCO framework, (ii) a reference governor, and (iii) a constraint tightening approach. We show that the proposed algorithm is recursively feasible and ensures robust constraint satisfaction. Furthermore, the algorithm's dynamic regret is bounded linearly in the variation of the cost functions and the magnitude of the disturbances.

Future work includes considering different combinations of OCO algorithms and reference governors than the OGD method \eqref{eq:OCO_OGD} and the scalar reference governor \eqref{eq:RG}. Furthermore, generalizing our approach to more general system classes, e.g., nonlinear systems or time-varying linear systems, is an interesting direction for future research.


\bibliographystyle{ieeetr}
\bibliography{bib}

\end{document}